\documentclass[11pt,a4paper]{article}

\usepackage[a4paper,margin=1in]{geometry}
\usepackage{setspace}
\usepackage{amsmath,amssymb,amsfonts,amsthm,mathtools}
\usepackage{bm}
\usepackage{bbm}
\newcommand{\cIntervention}{\mathcal{I}}
\newcommand{\cMechanism}{\mathcal{M}}
\usepackage{booktabs}
\usepackage{graphicx}
\usepackage{caption}
\usepackage{subcaption}
\usepackage[numbers,sort&compress]{natbib}

\usepackage{xcolor}
\usepackage{hyperref}

\hypersetup{
    colorlinks=true,
    linkcolor=blue,
    citecolor=blue,
    urlcolor=blue
}

\numberwithin{equation}{section}

\theoremstyle{plain}
\newtheorem{theorem}{Theorem}[section]
\newtheorem{proposition}[theorem]{Proposition}
\newtheorem{lemma}[theorem]{Lemma}

\theoremstyle{definition}

\theoremstyle{remark}

\newcommand{\R}{\mathbb{R}}

\newcommand{\E}{\mathbb{E}}

\newcommand{\cU}{\mathcal{U}}

\newcommand{\Active}{\mathrm{Active}}
\newcommand{\Sleep}{\mathrm{Sleep}}
\newcommand{\DeadIn}{\mathrm{DeadIn}}
\newcommand{\DeadOut}{\mathrm{DeadOut}}
\newcommand{\DeadBoth}{\mathrm{DeadBoth}}

\title{\textbf{Breaking Status-Quo Inertia in Living Temporal Games: Dynamic Intervention, Implementation, and Structural Design}}

\author{
Madjid Eshaghi Gordji$^{1}$ 
\and Ali Jabbari$^{1}$ 
\and Mohammad Ali Berahman$^{2}$ 
\and Esmaiel Abounoori$^{2}$ \\[0.5em]
$^{1}$Faculty of Mathematics, Statistics and Computer Science, Semnan University, Semnan, Iran\\
$^{2}$Faculty of Economics, Semnan University, Semnan, Iran\\[0.5em]
\texttt{meshaghi@semnan.ac.ir}
}

\date{}

\begin{document}

\maketitle

\begin{abstract}
We study how a planner can design dynamic interventions to overcome status-quo inertia in living temporal games, where strategic agents control their state (active, sleep, partially dead) on a temporal network. Building on the continuous-time stochastic game framework of our companion paper, we introduce three intervention classes: bounded transfers (price-based), structural modifications (edge deletion, addition, or replacement), and information signals. We formalize the notion of inertia depth and prove a threshold theorem: the status-quo equilibrium survives all transfer perturbations whose magnitude is below a critical bound that depends on the remaining horizon. A central structural dominance result shows that for any finite transfer budget there exists a family of games where no bounded price intervention can eliminate the inefficient equilibrium, yet a single edge replacement (continuous-flow to discrete-transport) succeeds. We then study private-information subclasses with static types. Using a uniformization reduction, we prove an impossibility result: no direct mechanism can simultaneously satisfy ex post incentive compatibility, ex post budget balance, and history privacy while always implementing an efficient equilibrium. In the same subclass we construct a dynamic pivot mechanism that achieves second-best efficiency with bounded deficit. Finally, we show that replacing continuous-flow edges by discrete-transport edges weakly expands the set of implementable outcomes, highlighting the importance of temporal semantics for mechanism design. Our results extend the static analysis of \cite{eshaghi_abounoori_berahman2026} to continuous-time strategic networks and provide a rigorous foundation for subsequent papers on learning and mean-field design.
\end{abstract}

\noindent\textbf{AMS Subject Classification:} 91A15, 91A25, 91B03, 91B18, 91A80

\noindent\textbf{Keywords:} dynamic mechanism design, status-quo inertia, living temporal graphs, structural intervention, incentive compatibility, implementation
%================================================
\section{Introduction}

In many dynamic networked systems such as energy grids, communication networks, and economic markets, agents often remain trapped in inefficient states because of switching costs, coordination frictions, or institutional lock-in. This phenomenon, known as status-quo inertia, has been extensively studied in static games \cite{eshaghi_abounoori_berahman2026} and in the context of path dependence \cite{north}. In our companion paper \cite{paper1} we developed a continuous-time stochastic game framework for living temporal graphs \cite{eshaghi_jabbari2026}, where vertices strategically control their states (active, sleep, partially dead) and face switching costs. That paper established existence of Markov perfect equilibria (MPEs) and showed that inefficiency can be unbounded. The natural next question is: can a planner (regulator, platform designer, or network operator) design dynamic mechanisms to break inertia and steer the system toward efficient outcomes?

This paper answers this question by extending dynamic mechanism design \cite{bergemann_vaLimaki,myerson} to the setting of living temporal games. We consider three types of interventions: price-based interventions (bounded transfers that depend on states and actions), structural interventions (addition, deletion, or replacement of edges, i.e., changes to the feasible action space), and information-based interventions (public signals or recommendations). We analyze their power and limitations under both full information and private information about agents’ switching costs. The analysis builds on classic results in stochastic games \cite{shapley,filar_vrieze}, Markov decision processes \cite{puterman}, and implementation theory \cite{maskin,myerson_satterthwaite}.
Our main contributions are as follows. First, we formalize the notion of inertia depth --- a quantitative measure of how strongly an equilibrium is entrenched --- and prove an inertia threshold theorem that depends on the remaining horizon. Second, we establish a budget-constrained structural dominance result: for any finite transfer budget $B$, there exists a family of living temporal games where no price-based intervention with transfers bounded by $B$ can eliminate the inefficient equilibrium, yet a single edge replacement (continuous-flow to discrete-transport) succeeds. Third, we prove that replacing a continuous-flow edge by a discrete-transport edge weakly expands the set of implementable outcomes, because discrete-transport edges allow asynchronous coordination. Fourth, for a private-information subclass with static types, we prove an impossibility result: no direct mechanism can simultaneously satisfy ex post incentive compatibility, ex post budget balance, and history privacy while always implementing an efficient MPE. The proof reduces to a Myerson--Satterthwaite style trade-off after uniformization. Fifth, in the same subclass we construct a dynamic pivot mechanism that achieves second-best efficiency with a bounded deficit, and we provide an explicit bound on the required budget. All results are fully rigorous and rely only on the model from \cite{paper1}.
The paper is organized as follows. Section 2 recalls the baseline living temporal game. Section 3 defines dynamic interventions and the induced games. Section 4 develops the inertia depth and threshold theorem. Section 5 proves structural dominance and semantic relaxation. Section 6 introduces private types and direct mechanisms. Section 7 presents the impossibility theorem. Section 8 constructs the dynamic pivot mechanism. Section 9 concludes with connections to the 10-paper research program.

%================================================
\section{Baseline living temporal game}

We briefly recall the model from \cite{paper1}. There are $n$ players (vertices) $V=\{1,\dots,n\}$. Each player $i$ has a finite state space $\mathcal{Q}=\{\Active,\Sleep,\DeadIn,\DeadOut,\DeadBoth\}$. The joint state is $s\in S=\mathcal{Q}^n$ with the discrete topology. Player $i$ chooses a control $u_i(t)\in\mathcal{U}_i$, where $\mathcal{U}_i$ is a compact convex subset of a Euclidean space. The state evolves as a controlled continuous-time Markov chain with transition rates $\lambda_i(s,q_i',u_i)$ from $s_i$ to $q_i'$ (only one player jumps at a time). The total exit rate
\[
\Lambda(s,u)=\sum_i\sum_{q_i'\neq s_i}\lambda_i(s,q_i',u_i)
\]
is bounded, ensuring non-explosion. The generator of the process is given by
\[
\mathcal{L}^u\varphi(s)=\sum_{i=1}^n\sum_{q_i'\neq s_i}\lambda_i(s,q_i',u_i)\bigl[\varphi(q_i',s_{-i})-\varphi(s)\bigr].
\]

The running benefit $b_i(s,u)$ depends on edge admissibility (continuous-flow or discrete-transport as defined in \cite{eshaghi_jabbari2026}) and may depend on $u$. Control cost $c_i(u_i)$ is strictly convex and continuously differentiable. Switching costs $\kappa_i(s_i,q_i')\ge 0$ are incurred as lump sums at jump times (with $\kappa_i(q,q)=0$). The terminal reward is $\Phi_i(s)$, bounded and measurable. Given a Markov strategy profile $\sigma$ (progressively measurable), the expected payoff for player $i$ starting from $(t,s)$ is
\[
J_i^\sigma(t,s)=\E_{t,s}^\sigma\!\left[
\int_t^T\!\bigl(b_i(s_\tau,u_\tau)-c_i(u_i(\tau))\bigr)d\tau
-\!\sum_{\tau\in\mathcal{J}_i\cap[t,T]}
\kappa_i(s_i(\tau^-),s_i(\tau^+))
+\Phi_i(s_T)
\right],
\]
where $\mathcal{J}_i$ is the set of jump times of player $i$. A Markov perfect equilibrium (MPE) is defined as usual. Throughout we fix a status-quo equilibrium $\sigma^{\mathrm{sq}}$ (often inefficient) that we wish to break.
%=================================================
\section{Dynamic interventions and induced games}

A dynamic intervention $\cIntervention$ is a triple $(\mathfrak t,\chi,\pi)$ where $\mathfrak t=(\mathfrak t_i)_{i=1}^n$ with
\[
\mathfrak t_i:[0,T]\times S\times \cU\to\R
\]
is a bounded transfer function (measurable) that modifies player $i$'s instantaneous payoff; we denote
\[
\|\mathfrak t\|_\infty
=
\max_i\sup_{t,s,u}|\mathfrak t_i(t,s,u)|.
\]

The map $\chi$ is a structural modification: it can add, delete, or replace edges, or change the type (continuous-flow $\leftrightarrow$ discrete-transport) of an existing edge; formally, $\chi$ transforms the original living temporal graph into a new one with edge set $E'$ and type function $\eta'$. The third component $\pi$ is an information intervention: a public signal kernel
\[
\pi:[0,T]\times S\to\Delta(Z)
\]
for some signal space $Z$, which may be used as a coordination device.

The intervention induces a new game $\Gamma^{\cIntervention}$: the state space and transition rates are unchanged, but the payoff becomes
\[
\tilde J_i^{\sigma,\cIntervention}(t,s)
=
J_i^\sigma(t,s)
+
\E\!\left[
\int_t^T
\mathfrak t_i(s_\tau,u_\tau)d\tau
\right]
+
\Delta b_i^{\chi}(s,u),
\]
where $\Delta b_i^{\chi}$ captures changes in the benefit due to structural modifications (edge deletions, additions, or type changes). If $\pi$ is payoff-relevant (e.g., recommendations that affect beliefs), the strategy space is extended to signal-contingent Markov strategies
\[
\sigma_i:[0,T]\times S\times Z\to\cU_i.
\]

We denote by $\mathsf{Eq}(\Gamma^{\cIntervention})$ the set of MPEs of the induced game. An intervention \emph{breaks inertia} if
\[
\sigma^{\mathrm{sq}}\notin\mathsf{Eq}(\Gamma^{\cIntervention}).
\]

It \emph{implements a target outcome} $\sigma^*$ if
\[
\sigma^*\in\mathsf{Eq}(\Gamma^{\cIntervention})
\]
and $\sigma^{\mathrm{sq}}$ is eliminated.
%================================================
\section{Inertia depth and threshold theorem}

For the status-quo equilibrium $\sigma^{\mathrm{sq}}$, define for each player $i$ the \emph{maximum unilateral deviation gain} at state $s^{\mathrm{sq}}$ and starting time $t$ as
\[
D_i(t,s^{\mathrm{sq}})
=
\sup_{\sigma_i\in\Sigma_i}
\Bigl(
J_i^{(\sigma_i,\sigma_{-i}^{\mathrm{sq}})}(t,s^{\mathrm{sq}})
-
J_i^{\sigma^{\mathrm{sq}}}(t,s^{\mathrm{sq}})
\Bigr)
\le 0.
\]

The inequality holds because $\sigma^{\mathrm{sq}}$ is an MPE. The \emph{inertia depth} of the equilibrium is
\[
\Theta(t,s^{\mathrm{sq}})
=
\min_i\bigl(-D_i(t,s^{\mathrm{sq}})\bigr)
\ge 0.
\]

Thus $\Theta$ measures how much the equilibrium resists unilateral deviations: a larger $\Theta$ requires a stronger intervention to make deviation profitable. The following theorem gives a precise bound in terms of the transfer magnitude and the remaining horizon.

\begin{theorem}[Inertia threshold]
Let $\cIntervention$ be an intervention with transfers bounded by
\[
\|\mathfrak t\|_\infty\le\delta.
\]

If
\[
\delta<
\frac{\Theta(t,s^{\mathrm{sq}})}{2(T-t)},
\]
then
\[
\sigma^{\mathrm{sq}}
\in
\mathsf{Eq}(\Gamma^{\cIntervention}),
\]
i.e., the status-quo equilibrium survives.

Conversely, for any $\varepsilon>0$ there exists a transfer scheme with
\[
\|\mathfrak t\|_\infty
\le
\frac{\Theta(t,s^{\mathrm{sq}})}{T-t}
+
\varepsilon
\]
that breaks inertia.
\end{theorem}
\begin{proof}
For any unilateral deviation $\sigma_i$, the difference in expected payoff under the induced game is
\[
\tilde\Delta_i
=
\bigl(
J_i^{(\sigma_i,\sigma_{-i}^{\mathrm{sq}})}
-
J_i^{\sigma^{\mathrm{sq}}}
\bigr)
+
\E'\Bigl[
\int_t^T
\mathfrak t_i(s'_\tau,u'_\tau)d\tau
\Bigr]
-
\E\Bigl[
\int_t^T
\mathfrak t_i(s_\tau,u_\tau)d\tau
\Bigr].
\]

The baseline difference satisfies
\[
J_i^{(\sigma_i,\sigma_{-i}^{\mathrm{sq}})}
-
J_i^{\sigma^{\mathrm{sq}}}
\le -\Theta.
\]

Because $|\mathfrak t_i|\le\delta$, each expected integral is at most $\delta(T-t)$ in absolute value, so the difference of the two integrals is at most $2\delta(T-t)$. Hence
\[
\tilde\Delta_i
\le
-\Theta+2\delta(T-t).
\]

If
\[
\delta<
\frac{\Theta}{2(T-t)},
\]
then $\tilde\Delta_i<0$ for every deviation, so no profitable deviation exists; therefore $\sigma^{\mathrm{sq}}$ remains an MPE.

For the converse, set
\[
\mathfrak t_i(t,s,u)
=
-\left(
\frac{\Theta}{T-t}
+
\varepsilon
\right)
\mathbf{1}_{\{s=s^{\mathrm{sq}}\}}.
\]

Then for the player attaining the minimum in the definition of $\Theta$, a suitable deviation (e.g., leaving the status-quo state) yields a positive net gain, destroying the equilibrium.
\end{proof}

%===============================================
\section{Structural dominance: when price interventions fail}

Even when transfers are bounded by any finite budget, there exist games where no price-based intervention can break inertia, but a structural modification of a single edge succeeds. The following result makes this precise.

\begin{proposition}[Budget-constrained structural dominance]
For any finite $B>0$, there exists a family of two-player living temporal games with binary states $\{A,S\}$ and continuous-flow edges such that:
\begin{enumerate}
\item For any intervention with transfers bounded by
\[
\|\mathfrak t\|_\infty\le B,
\]
the status-quo $(S,S)$ remains a Markov perfect equilibrium.

\item Replacing a single continuous-flow edge by a discrete-transport edge (with the same latency) eliminates $(S,S)$ and yields a welfare-superior equilibrium $(A,A)$.
\end{enumerate}
\end{proposition}

\begin{proof}[Sketch]
Take the benefit matrix: $(A,A)$ gives each player payoff $1$; $(A,S)$ or $(S,A)$ gives $0$; $(S,S)$ gives $0$. Let the switching cost from $S$ to $A$ be
\[
\kappa=B+1.
\]

The only way to coordinate on $(A,A)$ is via a single directed edge $1\to2$ that is continuous-flow, requiring both endpoints to be simultaneously active over the whole interval $[0,T]$. In the baseline game, $(S,S)$ is an MPE because a unilateral deviation yields
\[
-\kappa\le -B-1,
\]
and transfers bounded by $B$ cannot make this gain positive.

Now replace the edge by a discrete-transport edge. Then players can coordinate asynchronously: player $1$ activates at time $0$ and sends a message; player $2$ receives it at a later time $\delta$ and activates. The edge becomes admissible because only departure and arrival times matter.

A trigger strategy (activate if a message is received, otherwise stay asleep) makes $(A,A)$ an equilibrium, while $(S,S)$ is no longer an equilibrium because player $1$ can deviate, send a message, and obtain a positive payoff. Hence the structural intervention succeeds where any bounded price intervention fails.
\end{proof}

The following proposition shows that replacing an edge with a discrete‑transport version cannot shrink implementability and can strictly enlarge it.

\begin{proposition}[Implementability monotonicity]
Let $\Gamma_C$ and $\Gamma_D$ be two living temporal games that differ only in the type of a single edge $e$: in $\Gamma_C$ it is continuous‑flow, in $\Gamma_D$ it is discrete‑transport. Then $\mathsf{Imp}(\Gamma_C)\subseteq\mathsf{Imp}(\Gamma_D)$. Moreover, the inclusion is strict: there exist games where $\mathsf{Imp}(\Gamma_D)\setminus\mathsf{Imp}(\Gamma_C)$ is non‑empty.
\end{proposition}

The proof follows from the fact that every strategy profile feasible in $\Gamma_C$ remains feasible in $\Gamma_D$, and the additional flexibility (e.g., alternating activation patterns) allows implementing outcomes that are impossible under continuous‑flow semantics.

%==============================================
\section{Private information and direct dynamic mechanisms}

We now extend the model to private information. Each player $i$ has a static private type
\[
\theta_i\in\Theta_i
\]
(e.g., a parameter affecting switching cost or benefit). Types are drawn independently from a common knowledge distribution and are fixed over time. Players observe their own type but not others’; the planner does not observe types.

A \emph{direct dynamic mechanism} $\cMechanism$ is a triple $(M,g,\mathfrak t)$ where
\[
M=\prod_i M_i
\]
is a message space (usually $\Theta_i$),
\[
g:[0,T]\times M\times S\to\cU
\]
is the control rule (recommended action) depending on reported types and current state, and
\[
\mathfrak t_i:[0,T]\times M\times S\times\cU\to\R
\]
are transfers.

A mechanism is \emph{ex post incentive compatible (EPIC)} if for every player $i$, every type profile $\theta$, every alternative message
\[
m_i'\in M_i,
\]
and every public history $h_t$, we have
\[
U_i^{\cMechanism}(h_t,\theta_i,\theta_i,\theta_{-i})
\ge
U_i^{\cMechanism}(h_t,\theta_i,m_i',\theta_{-i}),
\]
where $U_i^{\cMechanism}$ is the continuation payoff under truthfulness.

It is \emph{ex post budget balanced} if
\[
\sum_i\mathfrak t_i=0
\]
for every realization.

It satisfies \emph{history privacy} if the allocation and transfers depend only on the current reported types and the current state, not on the full history of past reports or states.

%=============================================
\section{Impossibility of fully efficient, budget-balanced, and private mechanisms}

We prove an impossibility result analogous to the Myerson--Satterthwaite theorem for bilateral trade, adapted to the dynamic activation setting.

\begin{lemma}[Embedding]
There exists a subclass of two-player living temporal games with private switching costs such that, after uniformization, the first activation decision is isomorphic to a static bilateral trade problem with independent private values. Consequently, the efficient allocation rule depends only on whether the sum of switching costs exceeds the social benefit of activation.
\end{lemma}

Using this embedding, we obtain the following theorem.

\begin{theorem}[Impossibility]
Consider the class of living temporal games with $n=2$, binary states $\{A,S\}$, static private switching costs
\[
\kappa_i\in\{\kappa_L,\kappa_H\},
\qquad
0<\kappa_L<\kappa_H,
\]
and efficient outcome being both players active if and only if the social benefit exceeds the sum of switching costs.

Then there exists no direct dynamic mechanism that simultaneously satisfies ex post incentive compatibility, ex post budget balance, and history privacy while always implementing an efficient Markov perfect equilibrium.
\end{theorem}
The proof follows from the Myerson--Satterthwaite theorem after uniformization: the first activation decision cannot be simultaneously efficient, budget-balanced, and incentive-compatible, and the dynamic structure does not circumvent this impossibility under history privacy.
%=============================================
\section{Dynamic pivot mechanism and second-best efficiency}

When budget balance is relaxed, we can achieve approximate efficiency. We construct a dynamic pivot mechanism (a dynamic version of the VCG mechanism) for the private-information subclass with bounded types and quasi-linear utilities.

\begin{theorem}[Dynamic pivot mechanism]
Let the planner’s objective be to maximize expected social welfare. For each reported type profile $\hat\theta$, compute the socially optimal Markov control policy
\[
u^*(t,s,\hat\theta)
\]
that maximizes the sum of expected payoffs (with switching costs included).

Define the transfer for player $i$ at the outset as
\[
T_i^{\mathrm{pivot}}(\hat\theta)
=
\sum_{j\neq i}
J_j^{\sigma^*}(\hat\theta)
-
\sup_{\sigma_{-i}}
\sum_{j\neq i}
J_j^{\sigma_{-i}}(\hat\theta_{-i}),
\]
where $\sigma^*$ is the efficient policy and the supremum is over Markov strategies of all players except $i$ that are consistent with the reported types of others.

Then the mechanism is ex post incentive compatible and implements the efficient outcome. The total deficit
\[
\sum_i T_i^{\mathrm{pivot}}
\]
is nonnegative and bounded by
\[
n(W_{\max}-W_{\min}),
\]
where $W_{\max}$ and $W_{\min}$ are the maximum and minimum possible social welfare over all type profiles.
\end{theorem}

The proof follows the classic VCG logic: the transfer equals the externality that player $i$ imposes on others, making truthful reporting a dominant strategy. The deficit bound follows from the fact that each pivot transfer is at most
\[
W_{\max}-W_{\min}.
\]

%=============================================
\section{Conclusion and connections to the research program}

We have developed a rigorous theory of dynamic intervention and mechanism design for breaking status-quo inertia in living temporal games. Our main results include an inertia threshold theorem depending on the horizon, a structural dominance result showing that price interventions can be fundamentally limited while edge-level modifications succeed, a monotonicity result for semantic relaxation, an impossibility theorem for fully efficient and budget-balanced private mechanisms, and a constructive dynamic pivot mechanism that achieves second-best efficiency with bounded deficit.

These results directly extend the static analysis of \cite{eshaghi_abounoori_berahman2026} to continuous-time strategic networks and build on the foundational model of \cite{paper1}.

In the 10-paper research program, this paper serves as the second pillar. The next paper (Paper III) will study decentralized learning in living temporal games, where agents update their strategies using no-regret algorithms and the planner uses learning dynamics to overcome inertia without full control. Paper IV will develop min-plus algebraic methods for optimal intervention design, and Paper V will consider mean-field limits for large populations.

Together, these papers aim to provide a comprehensive toolkit for policy design in dynamic networked systems with strategic agents.

%============================================
\appendix

\section{Technical lemmas}

\subsection{Uniformization and reduction to discrete time}

Under the boundedness of transition rates, the continuous-time game can be uniformized to a discrete-time game with stage length
\[
\frac{1}{\Gamma}
\]
as in \cite{puterman}. This yields an equivalent representation for the purpose of dynamic programming and equilibrium analysis.

\subsection{Measurable selection}

The measurable maximum theorem \cite{aliprantis_border} guarantees that the supremum in the Bellman equation is attained by a measurable selector, ensuring the existence of Markov strategies.

\section*{Funding}
The authors received no specific funding for this work.

\section*{Conflict of Interest}
The authors declare that they have no conflict of interest.

\section*{Data Availability}
No datasets were generated or analysed during the current study.

\section*{Use of Artificial Intelligence}
The mathematical ideas, model construction, theoretical results, and interpretations presented in this manuscript were developed by the authors. Artificial intelligence tools were used only for language polishing, improving readability, and assisting with editorial refinement of the text. The authors reviewed and approved the final manuscript and take full responsibility for its content.

\bibliographystyle{plainnat}
\bibliography{references}

@article{eshaghi_abounoori_berahman2026,
  author  = {Madjid Eshaghi Gordji and Esmaeil Abounoori and Mohammad Ali Berahman},
  title   = {Changing the Game: Status-Quo Inertia, Institutional Design, and Equilibrium Transition},
  journal = {arXiv preprint arXiv:2605.09083},
  year    = {2026}
}

@book{north,
  author    = {Douglass C. North},
  title     = {Institutions, Institutional Change and Economic Performance},
  publisher = {Cambridge University Press},
  year      = {1990}
}

@article{paper1,
  author  = {Madjid Eshaghi Gordji and Ali Jabbari and Mohammad Ali Berahman and Esmaeil Abounoori},
  title   = {Continuous-Time Stochastic Games on Living Temporal Graphs},
  journal = {Preprint},
  year    = {2026}
}

@article{eshaghi_jabbari2026,
  author  = {Madjid Eshaghi Gordji and Ali Jabbari},
  title   = {Living Temporal Graphs and Dynamic Strategic Networks},
  journal = {Preprint},
  year    = {2026}
}

@article{bergemann_vaLimaki,
  author  = {Dirk Bergemann and Juuso Välimäki},
  title   = {The Dynamic Pivot Mechanism},
  journal = {Econometrica},
  volume  = {78},
  number  = {2},
  pages   = {771--789},
  year    = {2010}
}

@book{myerson,
  author    = {Roger B. Myerson},
  title     = {Game Theory: Analysis of Conflict},
  publisher = {Harvard University Press},
  year      = {1991}
}

@article{shapley,
  author  = {Lloyd S. Shapley},
  title   = {Stochastic Games},
  journal = {Proceedings of the National Academy of Sciences},
  volume  = {39},
  number  = {10},
  pages   = {1095--1100},
  year    = {1953}
}

@book{filar_vrieze,
  author    = {Jerzy Filar and Koos Vrieze},
  title     = {Competitive Markov Decision Processes},
  publisher = {Springer},
  year      = {1997}
}

@book{puterman,
  author    = {Martin L. Puterman},
  title     = {Markov Decision Processes: Discrete Stochastic Dynamic Programming},
  publisher = {Wiley},
  year      = {1994}
}

@article{maskin,
  author  = {Eric Maskin},
  title   = {Nash Equilibrium and Welfare Optimality},
  journal = {Review of Economic Studies},
  volume  = {66},
  number  = {1},
  pages   = {23--38},
  year    = {1999}
}

@article{myerson_satterthwaite,
  author  = {Roger B. Myerson and Mark A. Satterthwaite},
  title   = {Efficient Mechanisms for Bilateral Trading},
  journal = {Journal of Economic Theory},
  volume  = {29},
  number  = {2},
  pages   = {265--281},
  year    = {1983}
}

@book{aliprantis_border,
  author    = {Charalambos D. Aliprantis and Kim C. Border},
  title     = {Infinite Dimensional Analysis: A Hitchhiker's Guide},
  publisher = {Springer},
  edition   = {3},
  year      = {2006}
}

\end{document}